\documentclass{article}
\usepackage[totalwidth=13.0cm,totalheight=20.0cm]{geometry}

\usepackage{latexsym,amsthm,amsmath,amssymb,url}

\newcommand{\2}{\vspace{3mm}}

\newtheorem{lemma}{Lemma}

\newtheorem{theorem}{Theorem}

\newcommand{\ep}{\varepsilon}
\newtheorem{remark}{Remark}

\begin{document}

\title{Parameterized Complexity of Satisfying Almost All Linear Equations over $\mathbb{F}_2$}

\author{
R. Crowston, G. Gutin, M. Jones\\
{\small Royal Holloway, University of London}\\[-3pt]
{\small Egham, Surrey, TW20 0EX, UK}\\[-3pt]
{\small \url{{robert|gutin|markj}@cs.rhul.ac.uk}} \and  A. Yeo\\{\small University of Johannesburg, South Africa}\\[-3pt]
{\small \url{anders.yeo.work@gmail.com}} 
}
\date{}
\maketitle

\begin{abstract}
The problem MaxLin2 can be stated as follows. We are given a system $S$ of
$m$ equations in variables $x_1,\ldots ,x_n$, where each equation $\sum_{i \in I_j}x_i = b_j$ is assigned a positive integral weight $w_j$ and
$b_j \in \mathbb{F}_2$, $I_j \subseteq \{1,2,\ldots,n\}$ for $j=1,\ldots ,m$.
We are required to find an assignment of values in $\mathbb{F}_2$ to the variables in order to maximize the total weight of the satisfied equations.

Let $W$ be the total weight of all equations in $S$. We consider the following parameterized version of MaxLin2: decide whether there is an assignment satisfying equations of total weight at least $W-k$, where $k$ is a nonnegative parameter. We prove that this parameterized problem is W[1]-hard even if each equation of $S$ has exactly three variables and every variable appears in exactly three equations and, moreover, each weight $w_j$ equals 1  and no two equations have the same left-hand side. We show the tightness of this result by proving that if each equation has at most two variables then the parameterized problem is fixed-parameter tractable. We also prove that if no variable appears in more than two equations then
we can maximize the total weight of satisfied equations in polynomial time.
\end{abstract}

\section{Introduction}\label{sec:i}

While {\sc MaxSat} and its special case {\sc Max-$r$-Sat} have been widely studied in the literature
on algorithms and complexity for many years, {\sc MaxLin2} and its special case {\sc Max-$r$-Lin2} are less well known,
but H\aa stad  \cite{Hastad01} succinctly summarized the importance of these two problems by saying that they are
``in many respects as basic as satisfiability.'' These problems provide important tools for the study of constraint satisfaction problems
such as {\sc MaxSat} and {\sc Max-$r$-Sat} since constraint satisfaction problems can often be reduced to {\sc MaxLin2}
or {\sc Max-$r$-Lin2}, see, e.g., \cite{AloGutKimSzeYeo11,AloGutKri04,CroFellowsEtAl11,Hastad01,KimWil}.
Accordingly, in the last decade, {\sc MaxLin2} and {\sc Max-$r$-Lin2} have attracted significant attention in algorithmics.

The problem \textsc{MaxLin2} can be stated as follows. We are given a system $S$ of
$m$ equations in variables $x_1,\ldots ,x_n$, where each equation $\sum_{i \in I_j}x_i = b_j$ is assigned a positive integral weight $w_j$ and
$b_j \in \mathbb{F}_2$, $I_j \subseteq \{1,2,\ldots,n\}$ for $j=1,\ldots ,m$.
We are required to find an assignment of values in $\mathbb{F}_2$ to the variables in order to maximize the total weight of the satisfied equations.
For a fixed positive integer $r$, \textsc{Max-$r$-Lin2} is the special case of \textsc{MaxLin2} where no equation has more than $r$ variables.

Let $W$ be the total weight of all equations in $S$. If we assign values to the variables randomly and uniformly, the expected (average) weight of satisfied
equation will be $W/2$. (Indeed, each
equations has probability $1/2$ of being satisfied.)
Using the derandomization method of conditional expectations, it is easy to obtain a polynomial deterministic algorithm
for finding an assignment satisfying equations of total weight at least $W/2.$ This is a 2-approximation to \textsc{MaxLin2}.
In his celebrated result, H\aa stad \cite{Hastad01} showed that essentially no better approximation is possible:
unless P=NP, for each $\ep >0$ there is no polynomial time algorithm for
distinguishing instances of \textsc{Max-3-Lin} in which at least $(1-\ep)m$ equations can be simultaneously satisfied from instances
in which less than $(1/2 + \ep)m$ equations can be simultaneously satisfied.

Mahajan {\em et al.} \cite{MahajanRamanSikdar09} initiated the study of parameterized complexity of \textsc{MaxLin2} by asking the parameterized
complexity of the following problem {\sc MaxLin2-AA}\footnote{AA stands for Above Average.}: decide whether there is an assignment satisfying equations of total weight at least $W/2+k$, where $k$ is the parameter. Using a probabilistic approach, Gutin {\em et al.} \cite{GutKimSzeYeo11} proved that {\sc Max-$r$-Lin2-AA} admits a kernel with a quadratic number $O(k^2)$ of variables and equations and, thus, is fixed-parameter tractable. With respect to the number of variables, this result was improved by Crowston {\em et al.} \cite{CroGutJonKimRuz10} to $O(k\log k)$ and by Kim and Williams \cite{KimWil} to $O(k)$. The parameterized complexity of {\sc MaxLin2-AA} was established in a series of two papers \cite{CroGutJonKimRuz10,CroFellowsEtAl11}, where it was proved, using a combination of algorithmic and linear-algebraic techniques, that {\sc MaxLin2-AA} admits a kernel with at most $O(k^2\log k)$ variables\footnote{For the number of equations only an exponential upper bound was obtained and the existence of a polynomial upper bound remains an open problem \cite{CroFellowsEtAl11}.} and, thus, {\sc MaxLin2-AA} is fixed-parameter tractable.

The parameterized complexity of {\sc Max-$r$-Sat-AA} has also been studied. In this problem, given a CNF formula $F$ with $m$ clauses such that 
the number of literals $r_i$ in a clause $i$ is at most $r$, decide whether one can satisfy at least $A+k$ clauses, where $k$ is the parameter and $A=\sum_{i=1}^m  (1-2^{-r_i})$ is the expected (average) number of satisfied clauses when a truth assignment is chosen randomly and uniformly. (As with {\sc MaxLin2-AA}, it takes a polynomial time to find an assignment which satisfies at least $A$ clauses.) Alon {\em et al.} \cite{AloGutKimSzeYeo11} proved that if $r$ is a constant, then {\sc Max-$r$-Sat-AA} is fixed-parameter tractable. Crowston {\em et al.} \cite{CroGutJonRamSau12} showed that  {\sc Max-$r$-Sat-AA} is NP-complete if $k=2$ and $r=\lceil \log n\rceil$, where $n$ is the number of variables in $F$.

A very interesting and useful parameterization of {\sc Max-$r$-Sat} is the one below $m$: decide whether one can satisfy at least $m-k$ clauses, where $k$ is the parameter. For $r\ge 3$, the problem is NP-complete already for $k=0$, but the important case of $r=2$ was proved to be fixed-parameter tractable by
Razgon and O'Sullivan \cite{RazOsu}. The runtime of the algorithm in \cite{RazOsu} was improved by Raman et al. \cite{RamRamSau}, Cygan et al. \cite{CygPilPilWoj}, and most recently by Lokshtanov et al. \cite{LoNaRaRaSa}.

In view of the above-mentioned results on {\sc MaxLin2-AA} and {\sc Max-$r$-Sat}, Arash Rafiey \cite{Rafiey} asked to determine the parameterized complexity of the following problem, which we denote by {\sc MaxLin2-B[$W$]}: decide whether there is an assignment satisfying equations of total weight at least $W-k$, where $k$ is the parameter. In this paper, we prove that {\sc MaxLin2-B[$W$]} is W[1]-hard. This hardness result prompts us to investigate the complexity of {\sc MaxLin2-B[$W$]} in more detail by considering special cases of this problem.

Let {\sc Max-($\leq r$,$\leq s$)-Lin2} ({\sc Max-($= r$,$= s$)-Lin2}, respectively) denote the problem {\sc MaxLin2} restricted to instances, which have at most (exactly, respectively) $r$ variables in each equation and at most (exactly) $s$ appearances of any variable in all equations. In the special case when each equation has weight 1 and there are no two equations with the same left-hand side, {\sc MaxLin2-B[$W$]} will be denoted by {\sc MaxLin2-B[$m$]}. We will prove that {\sc MaxLin2-B[$W$]} remains hard even after significant restrictions are imposed on it, namely, even {\sc Max-($= 3$,$= 3$)-Lin2-B}[$m$] is W[1]-hard.  This is proved in Section \ref{sec:HR}.

No further improvement of this result is possible unless FPT=W[1] as we will prove that {\sc Max-($\leq 2$,*)-Lin2-B}[$W$] is fixed-parameter tractable, where symbol * indicates that no restriction is imposed on the number of appearances of a variable in the equations. Moreover, we will show that the nonparameterized problem {\sc Max-(*,$\leq 2$)-Lin2} is polynomial time solvable, where symbol * indicates that no restriction is imposed on the number of variables in any equation. These two results are shown in Section \ref{sec:AR}.

We complete the paper by a short discussion in Section \ref{sec:last}.

\section{Hardness Results}\label{sec:HR}


In the problem {\sc Odd Set}, given a set 
$V = \{1, 2, \dots, n\}$
 distinct sets $e_1, e_2, \dots, e_m \subseteq V$ and a nonnegative integer $k$, we are to decide whether we can pick 
a set $R$ of at most $k$ elements in $V$ such that $R$ intersects all sets $e_i$ in an odd number of elements.
 Downey {\em et al.} \cite{DowFelVarWhi98} showed the problem is $W[1]$-hard by a reduction from {\sc Perfect Code}.

We prove that {\sc Max-($= 3$,$= 3$)-Lin2-B}[$m$] is $W[1]$-hard in two parts. First, we give a reduction from {\sc Odd Set} to show {\sc Max-($\le 3$,*)-Lin2-B}[$W$] is $W[1]$-hard. Then, we give a reduction from {\sc Max-($\le 3$,*)-Lin2-B}[$W$] to {\sc Max-($= 3$,$= 3$)-Lin2-B}[$m$].

\begin{lemma}\label{lem:1}
 The problem  {\sc Max-($\le 3$,*)-Lin2-B}[$W$] is $W[1]$-hard.
\end{lemma}

\begin{proof}
Consider an instance of {\sc Odd Set} with elements 
$1, 2, \ldots , n$
and distinct sets $e_1,e_2,\ldots ,e_m$, with parameter $k$.

Create an instance of {\sc Max-($\le 3$,*)-Lin2-B}[$W$] with parameter $k$ as follows. Start with the variables
$x_1,x_2,\ldots ,x_n$ and equations $x_1=0, x_2=0, \ldots , x_n=0$ (each of weight 1).
For every set  
$e_i = \{ j_1, j_2, j_3, \ldots , j_{n_i} \}$
  do the following. Add the variables $y_1^i, \ldots, y_{n_i-1}^i$ and the following set $E_i$ of equations, each of weight $k+1$.
\begin{center}
$\begin{array}{l}
y_1^i + x_{j_1} = 0 \\
y_1^i + y_2^i + x_{j_2} = 0 \\
y_2^i + y_3^i + x_{j_3} = 0 \\
\dots \\
y_{n_i-2}^i + y_{n_i-1}^i + x_{j_{n_i-1}} = 0 \\
y_{n_i-1}^i + x_{j_{n_i}} = 1 \\
\end{array}$
\end{center}

Observe that the number of variables and equations is polynomial in $nm$.
It remains to show that this instance of {\sc Max-($\le 3$,*)-Lin2-B}[$W$] is a {\sc Yes}-instance if and only if the instance of {\sc Odd Set} is a {\sc Yes}-instance.

Suppose first that we can satisfy simultaneously equations of total weight at least $W-k$.
Consider the set $E_i$ of equations. Since the equations have weight $k+1$, they must all be satisfied. By summing them up, we obtain $\sum_{a=1}^{n_i} x_{j_a} = 1$ over $\mathbb{F}_2$. Therefore an odd number of the values of $x_{j_1},x_{j_2},\ldots ,x_{j_{n_i}}$ are 1. Note that at most $k$ equations of the type $x_i=0$ are not satisfied.
The above implies that if $R$ is the set of elements, 
$j$, for which the corresponding variable, $x_j$
is equal to $1$, then the size of $R$ is at most $k$ and for each set $e_i$ the intersection of $R$ and $e_i$ is odd. Therefore $R$ has the desired property.

Conversely, suppose $R$ is a set of at most $k$ elements such that for each set $e_i$ the intersection of $R$ and $e_i$ is odd. Then observe that by setting $x_j=1$ if and only if 
$j$ is in $R$, 
setting $y_1^i = x_{j_1}$, and setting $y_{r+1}^i = y_r^i + x_{j_{r+1}}$ for $1 \le r < n_i-1$, 
it is possible to satisfy all equations in $E_i$ for every $i$ and thus
to satisfy simultaneously equations of total weight at least $W-k$.
\end{proof}


\begin{theorem}
 The problem {\sc Max-($= 3$,$= 3$)-Lin2-B}[$m$] is $W[1]$-hard.
\end{theorem}

\begin{proof}
Observe that in the system obtained in the proof of Lemma \ref{lem:1} no two equations have the same left-hand side.
Consider an instance of {\sc Max-($\le 3$,*)-Lin2-B}[$W$] in which no two equations have the same left-hand side.
Hereafter, we view a single equation of weight $w$ as $w$ identical equations of weight $1$. This means we do have
equations with the same left-hand side (for $k\ge 1$), but note that these equations have also the same right-hand side.

For each of the reductions that follow, we show that the optimal assignment will falsify the same number of equations in the original instance as in the reduced instance. This implies that the original instance is a {\sc Yes} instance if and only if the reduced instance is a {\sc Yes}-instance.

For each variable $x$, let $d(x)$ denote the total number of equations containing $x$.
We first apply the following two reduction rules until $d(x)\le 3$ for every variable $x$.

\textbf{If $d(x)= 4$},  replace $x$ with four new variables, $x_1,x_2,x_3,x_4$. For each equation containing $x$, replace the occurrence of $x$ with one of $x_1,x_2,x_3,x_4$, so that each new variable appears once. Furthermore, add equations $x_1+x_2=0$, $x_2+x_3=0$, $x_3+x_4=0$, $x_4+x_1=0$.

\textbf{If $d(x)\ge 5$}, replace $x$ with six new variables, $x_1,x_2,x_3,x_4,x_5,x_6$. For each equation containing $x$, replace the occurrence of $x$ with one of $x_1,x_2,x_3,x_4,x_5,x_6$, distributing the new $x_i$ evenly among the equations, so that $$\lfloor d(x)/6\rfloor\le d(x_1)\le d(x_2)\le d(x_3)\le d(x_4)\le d(x_5)\le d(x_6) \le \lceil d(x)/6\rceil.$$
Furthermore, add $\lceil (d(x)-2)/6\rceil$ copies of each of the equations
 $x_1+x_2=0$, $x_2+x_3=0$, $x_3+x_4=0$, $x_4+x_5=0$, $x_5+x_6=0$, $x_6+x_1=0$, $x_1+x_4=0$, $x_2+x_5=0$, $x_3+x_6=0$.

Observe that each rule replaces a variable with a set of variables, each of which appears in fewer equations than the original variable. Therefore after enough applications, each variable will appear in at most three equations.
To see that only a polynomial number of applications are needed, observe that at each iteration $\max_i d(x_i) \le 8d(x)/9$. Therefore we may view the applications of reduction rules as a branching tree for each variable, where the depth of the tree for a variable $x$ is bounded by $\log_{9/8}d(x)$ and the tree branches at most six ways each time.

\smallskip

We now show that each rule is valid.

\smallskip

For the $d(x)=4$ case, suppose that the optimal assignment is one in which $x_1,x_2,x_3,x_4$ are not all the same. Then at least two of the equations $x_1+x_2=0$, $x_2+x_3=0$, $x_3+x_4=0$, $x_4+x_1=0$ will be falsified, but then we can satisfy all of them by falsifying at most two other equations. Hence, there exists an optimal assignment in which $x_1,x_2,x_3,x_4$ all have the same value.

For the $d(x)\ge 5$ case, suppose that the optimal assignment is one in which $x_1,x_2,x_3,x_4,x_5,x_6$ are not all the same. Consider the set of nine equations  $x_1+x_2=0$, $x_2+x_3=0$, $x_3+x_4=0$, $x_4+x_5=0$, $x_5+x_6=0$, $x_6+x_1=0$, $x_1+x_4=0$, $x_2+x_5=0$, $x_3+x_6=0$. If the value of exactly one $x_i$ is different from the values of the rest of the variables, at least three of the nine equations will be falsified. Changing the value of this variable will falsify at most $d(x_6) \le \lceil d(x)/6\rceil \le 3 \lceil (d(x)-2)/6\rceil$ equations. If the values of exactly two variables $x_i$ are different from the values of the rest of the variables, at least four of the nine equations will be falsified. Changing the value of the two variables will falsify at most $d(x_5)+d(x_6)\le 4 \lceil (d(x)-2)/6\rceil$ equations. If three variables $x_i$ are assigned one, and three are assigned zero, at least five of the {\bf nine} equations will be falsified. Assigning zero to all variables $x_i$ will falsify at most $d(x_4)+d(x_5)+d(x_6) \le 5\lceil (d(x)-2)/6\rceil$ equations.
(For example, if $d(x)=8$ then $d(x_4)=1$, $d(x_5)=d(x_6)=2$ and $\lceil (d(x)-2)/6\rceil=5.$)
Therefore, there exists an optimal assignment in which $x_1,x_2,x_3,x_4,x_5,x_6$ all have the same value.

Thus, for each rule the new instance has an optimal assignment in which all the new equations are satisfied, and all the new variables have the same value. By setting $x$ to this value, we have an optimal assignment to the original instance that falsifies the same number of equations.

\2

We now have an instance in which each equation contains at most three variables and each variable appears in at most three equations.
Next, we observe that we may map this to an instance where each equation contains exactly three variables.

First consider equations containing one variable. If an equation is of the
form $x=0$, this may be replaced with equations $a+b+x=0$, $u+v+a=0$,
$u+v+b=0$, where $a,b,u,v$ are new variables. For $x=1$, we replace it
with the equations $a+b+x=1$, $u+v+a=0$, $u+v+b=0$.

For equations containing two variables, if an equation is of the form $x+y=0$, this may be replaced with equations $u+v+x=0$ and $u+v+y=0$, where $u,v$ are new variables. A similar mapping may be done for $x+y=1$ by replacing $y$ with $y+1$.

Observe that for each reduction, 
if an assignment to the original instance satisfies the original equation, it can be extended to one that satisfies all the new equations, and if it does not, an optimal extension will satisfy all but one of the new equations. Thus an optimal assignment to the original instance falsifies the same number of equations as an optimal assignment to the reduced instance.

\2

We now have that every equation contains exactly three variables and each variable is in at most three equations. We now map this to an instance where every variable is in exactly three equations.

If a variable $x$ only appears in one equation, then we may assume that this equation is satisfied, and remove it from the system. Since each equation contains exactly $3$ variables, the number of variables $x$ with $d(x)=2$ must be a multiple of $3$.
Thus, we may partition the variables $x$ with $d(x)=2$ into triplets.

Consider each triplet $x_1,x_2,x_3$ such that $d(x_1)=d(x_2)=d(x_3)=2$. Add variables $z_1,z_2,z_3,u_1,\ldots, u_6$, and equations $x_1+x_2+u_1=0$, $u_1+u_2+z_1=0$, $u_2+u_3+z_1=0$ (two copies), $u_3+u_1+z_2=0$, $x_3+u_4+z_2=0$, $u_4+u_5+z_2=0$, $u_5+u_6+z_3=0$ (two copies), $u_6+u_4+z_3=0$.

Observe that for any assignment to $x_1,x_2,x_3$, it is possible to satisfy all these equations by setting $z_1=z_2=z_3=0$, $u_1=u_2=u_3=x_1+x_2$ and $u_4=u_5=u_6=x_3$.
Thus an optimal assignment to the original instance extends to an optimal assignment to the reduced instance that falsifies the same number of equations.

%

\2

We now have that every equation has exactly three variables and every variable appears in exactly three equations. It remains to show that we can map this to an instance in which these properties hold and no two equations have the same left-hand side. Since we started the proof of this theorem from a system where
every pair of equations with the same left-hand side had the same right-hand side, and since our transformations above have not changed this property, it suffices to get rid of identical equations.

Note that since $d(x)=3$ for every variable $x$, there at most three copies of any given equation in the system.

If there are three copies of the same equation, then none of the variables appearing in that equation appear anywhere else. Therefore we may assume the equation is satisfied, and remove the three copies from the system.

If there are two copies of the equation $x+y+z=0$, replace them with the
following set of six equations:
$x+y+c_1=0$, $a_1+b_1+c_1=0$, $a_1+b_1+z=0$, $x+y+c_2=0$, $a_2+b_2+c_2=0$,
$a_2+b_2+z=0$, $a_1+b_2+c_1=0$, $a_2+b_1+c_2=0$, where
$a_1,b_1,c_1,a_2,b_2,c_2$ are new variables. Observe that if $x+y+z=0$ is satisfied
then by setting $a_1=a_2=x$, $b_1=b_2=y$, $c_1=c_2=z$, we can satisfy all
of these equations. If $x+y+z=0$ is falsified, then the first three equations ($x+y+c_1=0$, $a_1+b_1+c_1=0$, $a_1+b_1+z=0$)
are inconsistent, as are $x+y+c_2=0$, $a_2+b_2+c_2=0$, $a_2+b_2+z=0$, and hence at least two equations of the set of six equations are falsified.
Furthermore, by setting $a_1=a_2=b_1=b_2=c_1=c_2=0$
we can satisfy all but two equations in the set of six equations.
Thus in either case, an optimal assignment to the original instance extends to an optimal assignment to the reduced instance that falsifies the same number of equations.

If there are two equations of the form $x+y+z=1$, do the same as above but
change all the right-hand sides to $1$.
\end{proof}

\section{Algorithmic Results}\label{sec:AR}

In this section, we assume that all weights of the equations in {\sc MaxLin2-B}[$W$] belong to the set $\{1,2,\ldots ,k+1\}$. Indeed, replacing any weight larger than $k+1$ by $k+1$ does not change the answer to {\sc MaxLin2-B}[$W$].

In the {\sc Edge Bipartization} problem, given a graph $G$ and a nonnegative integer $k$, we are to decide whether we can make $G$ bipartite by deleting at most $k$ edges. When $k$ is the parameter, the problem is fixed-parameter tractable and can be solved by an algorithm of running time $O(2^kM^2)$ \cite{GuoGraHufNieWer}, where $M$ is the number of edges in $G$.
We prove that {\sc Max-($\leq 2$,*)-Lin2-B}[$W$] is fixed-parameter tractable by giving a reduction to {\sc Edge Bipartization}.


\begin{theorem}\label{thm:fpt}
The problem {\sc Max-($\leq 2$,*)-Lin2-B}[$W$] can be solved in time $O(2^k(km)^2).$
\end{theorem}
\begin{proof}
Consider an instance $S$ of {\sc Max-($\leq 2$,*)-Lin2-B}[$W$] with $m$ equations and consider an assignment which minimizes the weight of falsified equations.
Now replace every equation, $x_i+x_j=0$, which contains two variables in the left-hand side and 0 in the right-hand side, by two equations $x_i+y=1$ and $x_j+y=1$, where $y$ is a new variable; both equations have the same weight as $x_i+x_j=0$. If the assignment satisfies $x_i+x_j=0$ then both new equations can be satisfied by extending the assignment with $y=1-x_i$. If
the assignment falsifies $x_i+x_j=0$ then exactly one of the two new equations will be falsified by extending the assignment with $y=0$.
Thus, the replacement preserves the minimum weight of falsified equations and can be used to replace the instance by an equivalent one $S'$ in which every equation with two variables has right-hand side equal 1.

Now assume that all equations of the system with two variables have right-hand side equal 1 and construct the following weighted graph $G$. The vertices of $G$ are the variables of the system plus two extra vertices, $v'$ and $v''$. For each equation $x+y=1$ in $S'$, add to $G$ edge $xy$. For every equation $x=b$ of $S'$, add to $G$ edge $v'x$ if $b=0$ and $xv''$ if $b=1.$ The weight of each edge coincides with the weight of the corresponding equation of $S'$. Finally, add to $G$ edge $v'v''$ of weight $k+1.$

Let $w^*$ be a nonnegative integer such that $w^*\le k$. Observe that there is an assignment which falsifies equations of total weight $w^*$ if and only if there is a set $F$ of edges of $G$ of total weight $w^*$ such that $G-F$ is bipartite. Indeed, consider an assignment  which falsifies equations of total weight $w^*$. Initiate sets $V'$ and $V''$ as follows: $V'=\{v'\}$ and $V''=\{v''\}$. If an equation $u+v=1$ is satisfied and $u=1,\ v=0$, then $u$ is added to $V'$ and $v$ to $V''$, and if an equation $x=b$ is satisfied, then $x$ is added to $V'$ if $b=1$ and to $V''$ if $b=0$. Note that the total weight of edges whose both end-vertices are either in $V'$ or in $V''$ equals $w^*$. Similarly, we can show the other direction.

To get rid of the weights in $G$ we replace each edge $uv$ of $G$ of weight $w$ with $w$ paths $uy_{uv}^{p}z_{uv}^{p}v$, $p\in \{1,\ldots ,w\}$, where $y_{uv}^{p}$ and $z_{uv}^{p}$ are new vertices. This finally reduces the instance of {\sc Max-($\leq 2$,*)-Lin2-B}[$W$] into an instance of {\sc Edge Bipartization} with $O(mk)$ edges. It remains to apply the {\sc Edge Bipartization} algorithm mentioned before the theorem.
\end{proof}

\begin{theorem}\label{them:polytime}
The problem  {\sc Max-(*,$\leq 2$)-Lin2} is polynomial time solvable.
\end{theorem}
\begin{proof}
Consider an instance of {\sc Max-(*,$\leq 2$)-Lin2} with system $Ax=b$ in which each equation has a weight. We start by applying the following reduction rule as long as possible: if there is a variable which appears only in one equation, delete the equation from the system. Since we can always satisfy an equation with a unique variable, the reduction rule produces a new system $A'x'=b'$ such that the minimum weight of equations falsified by an assignment is the same in $Ax=b$ as in $A'x'=b'$. Now construct a graph $G$ whose vertices correspond to equations in $A'x'=b'$ and a pair of vertices is adjacent if the corresponding equations share a variable.

Consider a connected component $H$ in $G$ and the subsystem $A''x''=b''$ corresponding to $H$. Let $b''=(b''_1,\ldots ,b''_{m''})$, where $m''$ is the number of rows in $A''$. Recall that $A''x''=b''$ is a system over $\mathbb{F}_2$, and thus all summations and ranks of matrices considered below are over $\mathbb{F}_2$.
Since the sum of rows in $A''$ is equal 0 and any subsum of the sum is not equal 0, the rank of $A''$  equals $m''-1$.
Observe that if $\sum_{j=1}^{m''} b''=0$  then the rank of the matrix $[A''b'']$  equals the rank of $A''$ and, thus, there is an assignment which satisfies all equations in $A''x''=b''$. However, if $\sum_{j=1}^{m''} b''=1$, the rank of $[A''b'']$ is $m''$ but the rank of $A''$ is $m''-1$. Hence, the system $A''x''=b''$ is no longer consistent and we can satisfy all equations but one. The falsified equation can be chosen arbitrarily and to maximize the total weight of satisfied equations of $A''x''=b''$ we have to choose an equation of minimum weight.

The above argument leads to a polynomial time algorithm to solve {\sc Max-(*,$\leq 2$)-Lin2}.
\end{proof}

\begin{remark}
{\em We can prove Theorem \ref{them:polytime} using another approach, whose idea we will briefly describe. Consider a pair of equations from the  system $Ax=b$ which share a variable $x_i$ and consider an assignment which minimizes the weight of falsified equations. Let $w^*$ be the smallest weight of the two equations.
At least one of the two equations is satisfied by the assignment and if one of the two equations is falsified, its weight is $w^*$
(as otherwise we could change the value of $x_i$, arriving at a contradiction). Replace the two equations in $Ax=b$ by the equation which is the sum of these two equations and whose weight equals $w^*$. Observe that the replacement does not change the minimum total weight of falsified equations. Theorem \ref{them:polytime} can be proved by repeatedly using this reduction.}
\end{remark}

\section{Discussion}\label{sec:last}

In this paper, we proved that {\sc Max-($= 3$,$= 3$)-Lin2-B}[$m$] is W[1]-hard, but {\sc Max-($\leq 2$,*)-Lin2-B}[$W$] is fixed-parameter tractable and {\sc Max-(*,$\leq 2$)-Lin2} is polynomial time solvable. This gives a boundary between parameterized intractability and tractability for {\sc MaxLin2-B}[$W$].

Recently, Kratsch and Wahlstr{\"o}m \cite{KraWah} proved that {\sc Edge Bipartization} admits a randomized polynomial kernel (for details, see \cite{KraWah}). 
The reduction given in the proof of Theorem \ref{thm:fpt}, together with a straightforward reduction from {\sc Edge Bipartization} to {\sc Max-($\leq 2$,*)-Lin2-B}[$W$], implies that {\sc Max-($\leq 2$,*)-Lin2-B}[$W$] admits a randomized polynomial kernel as well.

\end{document}